\documentclass[10pt]{article}

\usepackage{times}
\usepackage{setspace}
\usepackage{amsmath}
\usepackage{amssymb}
\usepackage{color}
\usepackage{mathdots}
\usepackage{tikz}
\usetikzlibrary{arrows}
\usetikzlibrary{shapes}
\usepackage{footnote}
\makesavenoteenv{tabular}
\makesavenoteenv{table}
\PassOptionsToPackage{boxed,section}{algorithm}
\usepackage{algorithm}
\usepackage{algorithmic}
\usepackage{hyperref}
\hypersetup{colorlinks=true}
\usepackage{color}

\newcommand{\wbsquare}{\tikz{\path[draw=black,fill=white] (0,0) rectangle (.08,.16);\path[draw=black,fill=black] (0.08,0) rectangle (.16,.16);}}
\newcommand{\wsquare}{\tikz{\path[draw=black,fill=white] (0,0) rectangle (.16,.16);}}
\newcommand{\bwsquare}{\tikz{\path[draw=black,fill=black] (0,0) rectangle (.08,.16);\path[draw=black,fill=white] (0.08,0) rectangle (.16,.16);}}
\newcommand{\bsquare}{\tikz{\path[draw=black,fill=black] (0,0) rectangle (.16,.16);}}

\newcommand{\lgwbsquare}{\tikz{\path[draw=black,fill=white] (0,0) rectangle (.1,.2);\path[draw=black,fill=black] (0.1,0) rectangle (.2,.2);}}
\newcommand{\lgwsquare}{\tikz{\path[draw=black,fill=white] (0,0) rectangle (.2,.2);}}
\newcommand{\lgbwsquare}{\tikz{\path[draw=black,fill=black] (0,0) rectangle (.1,.2);\path[draw=black,fill=white] (0.1,0) rectangle (.2,.2);}}
\newcommand{\lgbsquare}{\tikz{\path[draw=black,fill=black] (0,0) rectangle (.2,.2);}}

\providecommand{\lgray}[1]{\ \ensuremath{\overset{#1}{\wbsquare}}\ }
\providecommand{\rgray}[1]{\ \ensuremath{\overset{#1}{\bwsquare}}\ }
\providecommand{\white}[1]{\ \ensuremath{\overset{#1}{\wsquare}}\ }
\providecommand{\black}[1]{\ \ensuremath{\overset{#1}{\bsquare}}\ }

\newcommand{\res}[1]{\mbox{$\ \mid\!\!\!\frac{#1}{\
			\stackrel{\mbox{\scriptsize\it RES}}{\ }\ }\ $}}

\newcommand{\rores}[1]{\mbox{$\ \mid\!\!\!\frac{#1}{\
			\stackrel{\mbox{\scriptsize\it RO-Res}}{\ }\ }\ $}}

\newcommand{\varres}[1]{\mbox{$\ \mid\!\!\!\frac{#1}{\
			\stackrel{\mbox{\scriptsize\it Var-RO-Res}}{\ }\ }\ $}}
\newtheorem{definition}{\em Definition}[section]
\newenvironment{proof}{ {\bf Proof:}} {$\Box$}
\newtheorem{lemma}{\em Lemma}[section]
\newtheorem{theorem}{\em Theorem}[section]
\newtheorem{corollary}{\em Corollary}[section]
\newcounter{excounter}
\newenvironment{example}{\stepcounter{excounter}{\em Example 
(\arabic{excounter}):}
\begin{normalsize}}{\end{normalsize}}
\long\def\comment#1{}
\binoppenalty=\maxdimen
\relpenalty=\maxdimen
\title{On the computational complexity of read once resolution decidability in $2$CNF formulas}

\author{Hans Kleine B\"uning\\ Computer Science Institute, \\ University of Paderborn \\ kbcsl@uni-paderborn.de \and Piotr Wojciechowski\thanks{This research was supported in part by the National
Science Foundation through Award CCF-1305054.}\\ LDCSEE \\ West Virginia University \\ pwjociec@mix.wvu.edu \and K. Subramani\thanks{This work was supported by the Air Force Research Laboratory under US Air Force contract FA8750-16-3-6003. The views expressed are those of the authors and do not reflect the official policy or position of the Department of Defense or the U.S. Government.}\\ LDCSEE \\ West Virginia University \\ k.subramani@mail.wvu.edu}

\begin{document}





\maketitle

\begin{abstract}
   In this paper, we analyze $2$CNF formulas from the perspectives of Read-Once resolution  (ROR) refutation schemes. We focus on two types of ROR refutations, viz.,
   variable-once  refutation and clause-once refutation. In the former, each variable may be used at most once in the derivation of a refutation, while in the latter, each clause may
   be used at most once. We show that the problem of checking whether a given $2$CNF formula has an ROR refutation under both schemes is
    {\bf NP-complete}. This is surprising in light of the fact that there exist polynomial refutation  schemes (tree-resolution and DAG-resolution) for $2$CNF formulas.
    On the positive side, we show that $2$CNF formulas have copy-complexity $2$, which means that any unsatisfiable $2$CNF formula has a refutation in
    which any clause needs to be used at most twice.

\end{abstract}

\section{Introduction}

   Resolution is a refutation procedure that
   was introduced in \cite{Robinson1} to establish the unsatisfiability of clausal
   boolean formulas. Resolution is a sound and complete procedure, although it
   is not  efficient  in general \cite{Hak1}.
    Resolution is one among many proof systems (refutation systems) that have
      been discussed in the literature \cite{Urqu1}; indeed it is among the weaker proof
      systems \cite{Bus2} in that there exist propositional formulas for which short proofs exist
      (in powerful proof systems) but resolution proofs of unsatisfiability are
      exponentially long. Resolution remains an attractive option for studying the
      complexity of constraint classes on account of its simplicity and wide applicability;
      it is important to note that resolution is the backbone of a range of automated
      theorem provers \cite{Bp2}.

    Resolution refutation techniques often arise in proof complexity.  Research in proof complexity is primarily concerned with the
     establishment of non-trivial lower bounds on the proof lengths
     of propositional tautologies (alternatively refutation lengths of
     propositional contradictions). An essential aspect of
     establishing a lower bound is the proof system used to establish
     the bound. For instance, super-polynomial bounds for tautologies
     have been established for weak proof systems such as resolution
     \cite{Hak1}.
     Establishing that there exist short refutations for all
     contradictions in a given proof system causes the classes {\bf
     NP} and {\bf coNP} to coincide \cite{CookReck1}. 

     There are a number of different types of resolution refutation that have been discussed in 
      the literature \cite{har1}. The most important types of resolution refutation are {\em tree-like, dag-like and read-once}.  Each type of resolution is characterized by a restriction on input clause combination.
      One of the simplest types of resolution is Read-once Resolution (ROR). In an ROR refutation, each input clause
      and each derived clause  may be used at most once. There are several reasons to prefer a ROR proof over a generalized resolution
      proof, not the least of which is that ROR proofs must {\em necessarily} be of length polynomial (actually, linear)
      in the size of the input. It follows that ROR cannot be a complete proof system unless {\bf NP = coNP}.
      That does not preclude the possibility that we could check in polynomial time whether or not a given
      CNF formula has a ROR refutation. Iwama \cite{IM95} showed that even in case of
      $3$CNF formulas, the problem of checking ROR existence (henceforth, ROR decidability) is {\bf NP-complete}.

       It is well-known that $2$CNF satisfiability is decidable in polynomial time. There are several algorithms for
       $2$CNF satisfiability, most of which convert the clausal formula into a directed graph and then exploit
       the connection between the existence of labeled paths in the digraph and 
       the satisfiability of the input formula.\comment{ Not surprisingly resolution refutation proofs for $2$CNF formulas
       are also polynomial in the size of the input. It was recently established that both the tree-like and dag-like
       resolution refutations of $2$CNF formulas can be computed in polynomial time.} A natural progression of
       this research is to establish the ROR complexity of $2$CNF formulas. We show that the problem of
       deciding whether an arbitrary $2$CNF formula has a read-once refutation is {\bf NP-complete}. Although ROR is an {\bf incomplete}
       refutation technique, we may be able to find a refutation if clauses can be copied. We show that
       every $2$CNF formula has an ROR refutation, if every clause can be copied once.

\smallskip

    The principal contributions of this paper are as follows:
 \begin{enumerate}
 \item Establishing the computational complexity of ROR decidability for $2$CNF formulas.
 \item Establishing the clause copy complexity for ROR existence in $2$CNF formulas.
 \end{enumerate}
 
   The rest of this paper is organized as follows: In Section \ref{sop}, we discuss problem preliminaries and formally define the various types of refutations
   discussed in this paper. The minimal unsatisfiable subset problems is detailed in Section \ref{minunsat}. In Section \ref{vror}, the Variable
   Read Once Resolution (VAR-ROR) refutation problem is detailed. We also establish the computational complexity of this problem in $2$CNF.
   We show that ROR decidability for $2$CNF formulas is {\bf NP-complete} in Section \ref{for}. The copy complexity of $2$CNF formulas is 
   established in Section \ref{copy}. Finally, we conclude in Section \ref{conc}, by summarizing our contributions and outlining avenues for future research.

\section{Preliminaries}
\label{sop}
 In this section, we briefly discuss the terms used in this paper.
 We assume that the reader is familiar with elementary propositional logic.
 A literal is a variable $x$ or its complement $\neg x$. $x$ is termed a positive and $\neg x$ is termed a negative literal. A clause is a disjunction of literals. The empty clause, which is always false, is denoted as $\sqcup$.\\
 A Boolean formula $\Phi$ is in CNF, if the formula is a conjunction of clauses. Please note, that a CNF is a set of clauses and written as
 $\{\alpha_1, \ldots, \alpha_n\}$, $\alpha_1 \wedge \ldots \wedge \alpha_n$, or simply as $\Phi= \alpha_1, \ldots, \alpha_n$ for clauses $ \alpha_i$.
A formula in CNF is in $k$-CNF, if it is of the form
 $\alpha_{1} \wedge \alpha_{2} \wedge \ldots \wedge \alpha_{m}$, where each $\alpha_{i}$ is a clause of at most $k$ literals. \comment{A Horn formula is a formula in CNF for which every clause contains at most one positive literal.}

 For a single resolution step with parent clauses $(\alpha \vee x)$ and 
 $(\neg x \vee \beta)$ with resolvent $(\alpha \vee \beta)$, we write
\[(\alpha \vee x),(\neg x \vee \beta) \res{1} (\alpha \vee \beta).\]
The variable $x$ is termed matching or resolution variable.
 If for initial clauses $\alpha_1, \ldots, \alpha_n$, a clause $\pi$ can be
 generated by a sequence of resolution steps we write 
 \[\alpha_1, \ldots, \alpha_n \res{} \pi.\]

    We now formally define the types of resolution refutation discussed in this paper.
\comment{
\begin{definition}
A {\em Tree-Like} resolution refutation is a refutation in which each derived clause can be used at most once. 
\end{definition}

Note that in tree-like refutations, the input clauses can be used multiple times and thus any derived clause
can be derived multiple times. Tree-like refutation is a {\bf complete} refutation procedure \cite{Bp2}.


\begin{definition}
A {\em Dag-Like} resolution refutation is a refutation in which each clause can be used multiple times. 
\end{definition}

 It follows that Dag-like refutations procedures are {\bf complete} as well. Furthermore Dag-like refutations $p$-simulate
 tree-like refutations \cite{CookReck1}.
 }

\begin{definition}
A formula is in {\bf Var-ROR} (variable-read once resolution), if and only if there is a resolution refutation for which every variable is used at most once as a matching variable.
\end{definition}

A resolution derivation $\Phi \res{} \pi$ is a Var-ROR derivation, if the matching variables are used at most once. We denote this as $\Phi \varres{} \pi$.
\begin{definition}
A formula $\Phi$ is said to be {\bf minimally Var-ROR}, if and only if 
$\Phi \in $ Var-ROR and every proper sub-formula is not in Var-ROR.
\end{definition}

\begin{definition}
A {\bf Read-Once} resolution refutation is a refutation in which each clause, $\pi$, can be used in only one resolution step. This applies to clauses present in the original formula and those derived as a result of previous resolution steps.
\end{definition}

A resolution derivation $\Phi \res{} \pi$ is a read-once resolution derivation, if 
for all resolution steps $\pi_1 \wedge \pi_2 \res{1} \pi$, we delete the clauses $\pi_1$ and $\pi_2$ from, and add the resolvent $\pi$ to,	the current set of clauses. In other words, if $U$ is the current set of clauses, we obtain
$U = (U \setminus \{\pi_1, .\pi_2\}) \cup \{\pi\}$. An example of this refutation method can be found in Appendix \ref{refexamp}.

ROR is the set of formulas in CNF, for which a read-once resolution refutation exists $(\Phi \in$ ROR if and only if $\Phi \rores{} \sqcup)$.

\begin{example}
Consider the following $2$CNF formula:
\begin{eqnarray*}
(x_1,x_2) & (x_3,x_4) & (\neg x_1, \neg x_3)\\
(\neg x_1,\neg x_4) & (\neg x_2,\neg x_3) & (\neg x_2,\neg x_4)
\end{eqnarray*}
We now show that this formula does not have a read-once refutation.

To derive $(x_1)$ we need to derive $(\neg x_2)$. Similarly, to derive $(x_2)$ we need to derive $(x_1)$. However, the derivations of both $(\neg x_1)$
and $(\neg x_2)$ require the use of the clause $(x_3,x_4)$.

To derive $(x_3)$ we need to derive $(\neg x_4)$. Similarly, to derive $(x_4)$ we need to derive $(\neg x_3)$. However, the derivations of both $(\neg x_3)$ and $(\neg x_4)$ require the use of the clause $(x_1,x_2)$.
\end{example}

\begin{definition}
A formula, $\Phi$, is {\bf minimally ROR} if and only if the formula is in ROR and 
every proper sub-formula is not in ROR.
\end{definition}

It is important to note that both types of  Read-Once resolution (Var-ROR and ROR) are  {\bf incomplete} refutation procedures.
Furthermore, ROR is a strictly more powerful refutation procedure than VAR-ROR. For instance, consider the formula:
\[\alpha= \{(a \vee b), (a  \vee \neg b), (\neg a \vee b), (\neg a \vee \neg b)\}. \]
It is not  hard to see that $\alpha$ has a ROR refutation
but not Var-ROR refutation.

After a resolution step $(\pi_1 \vee x) \wedge (\neg x \vee \pi_2) \res{1} (\pi_1 \vee \pi_2)$ in a Var-ROR refutation no clause with a literal over $x$ can be used. Thus, we can delete all clauses which include the variable $x$. We see that Var-ROR $\subseteq$ ROR. For 2-CNF formulas, Var-ROR is a proper subset of ROR. 

\begin{example}
Let $\alpha= \{(a \vee b), (a  \vee \neg b), (\neg a \vee b), (\neg a \vee \neg b)\}$.
We have that $\alpha \in$ ROR, as shown by the following series of refutation steps:
\begin{enumerate}
\item $(a \vee b) \wedge (\neg a \vee b) \res{1} (b)$.
\item $(a \vee \neg b) \wedge (\neg a \vee \neg b) \res{1} \neg (b)$.
\item $(b) \wedge (\neg b) \res{1} \sqcup$.
\end{enumerate}
However, $\alpha$ is not in Var-ROR. After the first resolution step, $(a \vee b) \wedge (\neg a \vee b) \res{1} (b)$, the variable $a$ can not be used as the matching variable in any other resolution step. Thus, we cannot derive $(\neg b)$. Hence there is no Var-ROR refutation.
\end{example}
\comment{
\subsection{ROR}
Under read-once refutation once a clause has been used in a resolution it can be considered removed from the formula. Thus, any clause in the original formula can be used only once. Note that, if a clause can be generated using two disjoint sets of original clauses then that clause can be reused as long as both derivations of that clause are present in the refutation. An example of this refutation method can be found in Appendix \ref{appROR}.

\subsection{Tree-like Resolution}
Under tree-like refutation we can reuse the clauses in the original formula. However, we cannot reuse derived clauses without rederiving them. Each time a clause is rederived it contributes to the size of the refutation. This occurs even if the clause is rederived from the same clauses originally present in the formula. Additionally, each time one of the original clauses is reused it also contributes to the size of the refutation. An example of this resolution method can be found in Appendix \ref{appTLR}.

\subsection{Dag-like Resolution}
Under dag-like refutation we can reuse any clause (original or derived). Unlike for a tree-like refutation, each time a clause is reused it does not contribute to the size of the refutation. Since tree style refutations allow for clauses to be rederived, every formula that has a DAG-like refutation has a tree style refutation. However, the DAG-like refutation is usually smaller.  An example of this resolution method can be found in Appendix \ref{appDLR}.
}
\section{Minimal Unsatisfiability}
\label{minunsat}

In the characterization of various read-once classes and in the proofs, we make use of minimal unsatisfiable formulas and splitting of these formulas.

First, we recall some notions and results.
\begin{definition}
	A formula in CNF is {\bf minimal unsatisfiable}, if and only if the formula is unsatisfiable and every proper sub-formula is satisfiable.
	The set of minimal unsatisfiable formulas is denoted as MU.
\end{definition}

\begin{definition}
	The {\bf deficiency} of a formula $\Phi$, written as $d(\Phi)$, is the number of clauses minus the number of variables. For fixed $k$, MU($k$) 
	is the set of MU-formulas with deficiency $k$.
\end{definition}

The problem of deciding whether a formula is minimal unsatisfiable is ${\bf D^P}$-{\bf complete} \cite{PW88}. ${\bf D^P}$ is the class of problems which can be
represented as the difference of two {\bf NP}-problems. Every minimal unsatisfiable formula has a deficiency greater or equal than 1 \cite{AL86}.

For fixed $k$, deciding if $\Phi \in$ MU($k$) can be solved in polynomial time {\cite{FKS02}. For formulas in $2$CNF, there is no constant upper bound for the deficiency of minimal unsatisfiable 
formulas.\comment{, whereas minimal unsatisfiable Horn formulas always have deficiency 1.}

The proofs in this paper make use of so-called splitting formulas for MU-formulas.

\begin{definition}
	Let 
	\begin{eqnarray*}
	\Phi &= & (x \vee \pi_1), \ldots, (x \vee \pi_r), \sigma_1, \ldots, \sigma_t,\\
	& & (\neg x \vee \phi_1), \ldots, (\neg x \vee \phi_q)
	\end{eqnarray*}
	 be a minimal 
	unsatisfiable formula, where neither the literal $x$ nor the literal $\neg x$ occur in the clauses $\sigma_i$. 
	A pair of formulas $(F_x, F_{\neg x})$ with $F_x = \pi_1, \ldots, \pi_r, \sigma_{i_1}, \ldots, \sigma_{i_s}$ and 
	$F_{\neg x} = \sigma_{j_1}, \ldots, \sigma_{j_k}, \phi_1, \ldots, \phi_q$ is called a {\bf splitting} of $\Phi$ over $x$, if $F_x$ and $F_{\neg x}$ are minimal
	unsatisfiable.
\end{definition}
\begin{definition}
	A splitting $(F_x, F_{\neg x})$ is {\bf disjunctive}, if $\{\sigma_{i_1}, \ldots, \sigma_{i_s}\} \cap
	\{\sigma_{j_1}, \ldots, \sigma_{j_k}\}$ is empty. That is, $F_x$ and $F_{\neg x}$ have no clause $\sigma_i$ in common.
	If additionally $F_x$ and $F_{\neg x}$ do not share any variables, 
	then we say that the splitting is {\bf variable-disjunctive}.
\end{definition}

We can continue to split both $F_x$ and $F_{\neg x}$ to obtain a splitting tree. Splitting stops when the formula contains only one variable.
\begin{definition}
A splitting tree is {\bf complete} if every leaf of the tree is a formula that contains at most one variable.
\end{definition}

Note that, after splitting on the variable $x$, neither $F_x$ nor $F_{\neg x}$ contains the variable $x$. Thus, the splitting tree can have depth at most $(n-1)$.

  In case of disjunctive splittings, we speak about disjunctive splitting trees.
It is known that formulas in MU($1$) have variable-disjunctive splitting trees \cite{DDKB98}.
Moreover, every minimal unsatisfiable formula with a read-once resolution refutation has a disjunctive splitting tree and vice versa. \cite{KBZ02}.

Let $\Phi$ be a minimal unsatisfiable $2$CNF formula. We will now prove several properties of $\Phi$.

\begin{lemma}
\label{lem01}
 If $\, \Phi$ contains a unit clause, then $\Phi \in$ MU(1).
\end{lemma}

\begin{proof}
	(By induction on the number $n$ of variables) If $\Phi$ has only one variable, then $\Phi = (x) \wedge (\neg x)$. Clearly, $\Phi \in$ MU($1$).
	
	Now assume that $\Phi$ has $(n+1)$ variables. Let $(x) \in \Phi$ be a unit clause. Thus, $\Phi$ has the form 
	\[\Phi = (x) \wedge (\neg x \vee L_1) \wedge \ldots \wedge (\neg x \vee L_k) \wedge \sigma,\]
	where $x$ and $\neg x$ do not occur in any clause in $\sigma$. Thus, the formula
	\[\Phi_x = (L_1) \wedge \ldots \wedge (L_k) \wedge \sigma\]
	 is minimal unsatisfiable. We have that $(L_1), \ldots, (L_k)$ are unit clauses. Thus, by the induction
	hypothesis $\Phi_x \in$ MU($1$). This means that $\Phi \in$ MU($1$), since $\Phi$ contains one more clause and one more variable than $\Phi_x$.
\end{proof}

\begin{lemma}
\label{lem02}
	For every variable $x$, the splitting formulas $F_x, F_{\neg x}$ are in MU($1$).
\end{lemma}

\begin{proof}
	We have that the splitting formulas $F_x$ and $F_{\neg x}$ contain a unit clause and are minimal unsatisfiable. Thus, from lemma \ref{lem01}, $F_x, F_{\neg x}$ are in MU($1$).
\end{proof}

\begin{lemma} \cite{KBZ03}
\label{lem03}
	If for a variable $x$, there is a disjunctive splitting, then the splitting is unique.
\end{lemma}

\begin{lemma}
\label{lem04}
 The problem of determining if $\Phi$ has a complete disjunctive splitting is in {\bf P}. Furthermore, if $\Phi$ has a disjunctive splitting for some variable $x$,
 then there exists a complete splitting tree.
\end{lemma}

\begin{proof}
	Let $\Phi = (x \vee L_1) \wedge (x \vee L_k) \wedge \sigma \wedge (\neg x \vee K_1) \wedge (\neg x \vee K_r)$. For $x$, we  take the formula
	$(L_1) \wedge (L_k) \wedge \sigma$ and reduce the formula to a minimal unsatisfiable formula $F_x= (L_1) \wedge (L_k) \wedge \sigma_1$ where 
	$\sigma_1 \subseteq \sigma$. This can be done in polynomial time, because the satisfiability problem for 2-CNF is solvable in linear time. Similarly, for $\neg x$, we can compute $F_{\neg x}$.
	 
	To decide whether the formula has a disjunctive splitting tree, it suffices to look for a variable $x$ for which $(F_x,F_{\neg x})$ is a disjunctive splitting.
	From lemma \ref{lem02}, $F_x$ and $F_{\neg x}$ are in MU(1) and therefore have disjunctive splitting trees \cite{DDKB98}.
	
	For each $x$, we compute a splitting $(F_x,F_{\neg x})$. If this splitting is disjunctive, we have found the desired splitting. If it is not disjunctive, then for $x$ there is no disjunctive splitting. If $x$ did have a disjunctive splitting, then, by lemma \ref{lem03}, that splitting would be unique meaning that $x$ would have no non-disjunctive splittings. 
	
	Hence, the problem is solvable in polynomial time.
\end{proof}

\section{The Complexity of Var-ROR for 2CNF}
\label{vror}

In this section, we show that determining if a formula in $2$CNF has a variable read-once refutation is {\bf NP-complete}. The Var-ROR problem is {\bf NP-complete} for CNF formulas in general \cite{Sze01}. We show that this result holds even when restricted to $2$CNF formulas.

Let $\Phi$ be a formula in $2$CNF.

\begin{theorem}
	\label{Theorem 4.1}
	$\Phi \in$ Var-ROR, if and only if there exists a sub-formula $\Phi' \subseteq \Phi$ such that $\Phi' \in$ MU($1$).
	
\end{theorem}
Theorem \ref{Theorem 4.1} follows immediately from lemma \ref{lemma4.1}.

\begin{lemma}
	\label{lemma4.1}
	$\Phi$ is minimally Var-ROR, if and only if $\Phi \in $ MU($1$).	
\end{lemma}

\begin{proof} 
Let $\Phi$ be minimally Var-ROR. We will show that $\Phi \in $ MU($1$) by induction on the number of variables.

	First, assume that $\Phi$ is a CNF formula over $1$ variable. Thus, $\Phi$ has the form $(x) \wedge (\neg x)$. Obviously, $\Phi \in$ MU($1$).
	
	Now assume that $\Phi$ is a CNF formula over $(n+1)$ variables. Let 
	\[(\alpha \vee x), (\neg x \vee \beta) \res{1} (\alpha \vee \beta)\]
	 be the first resolution step in a Var-ROR refutation of $\Phi$. Thus, the formula 
	\[\Phi'= (\Phi \setminus \{(\alpha \vee x), (\neg x \vee \beta)\} \cup \{(\alpha \vee \beta)\})\]
	is minimally Var-ROR
	and contains no clause with $x$ or $\neg x$. $x$ has already been used as a matching variable. Thus, if any other clauses of $\Phi$, or $\Phi'$, used $x$ then $\Phi$ would not be minimally Var-ROR.
	
	$\Phi'$ has $n$ variables. Thus, by the induction hypothesis $\Phi' \in $ MU(1). Since $\Phi'$ is minimal unsatisfiable and consists of $(n+1)$ clauses, $\Phi$ is minimal unsatisfiable and consists of $(n+2)$ clauses. This means that $\Phi \in $ MU(1).
		
Now let $\Phi$ be a formula in MU($1$). We will show that $\Phi$ is minimally Var-ROR by induction on the number of variables.

	For every formula in MU(1), there exists a variable-disjunctive splitting tree \cite{DDKB98}. Thus, we can easily construct
	a Var-ROR refutation for $\Phi$.
	
	First, assume that $\Phi$ is a CNF formula over $1$ variable. Thus,  $\Phi = (x) \wedge (\neg x)$ and is minimally Var-ROR.
	
	Now assume that $\Phi$ is a CNF formula over $(n+1)$ variables. Let $(F_x, F_{\neg x})$ be the first variable-disjunctive splitting in a variable-disjunctive splitting tree.
	Without loss of generality we assume that neither $F_x$ nor $F_{\neg x}$ is the empty clause. By the induction hypothesis, both formulas  are minimally Var-ROR, because, by lemma \ref{lem02}, $F_x,F_{\neg x} \in$ MU($1$).
	
	Thus, there is a Var-ROR derivation $F_x^x \varres{} (x)$ and $F_{\neg x}^{\neg x} \varres{} (\neg x)$, where $F_x^x$ ($F_{\neg x}^{\neg x}$) is the formula 
	we obtain by adding the removed literal $x$ ($\neg x$) to the clauses in $F_x$ ($F_{\neg x}$). The final step is to resolve $(x)$ and $(\neg x)$. Note that the variable $x$
	has not been used as a matching variable in the derivations $F_x^x \varres{} (x)$ and $F_{\neg x}^{\neg x} \varres{} (\neg x)$. Hence, $\Phi$ is in Var-ROR and is minimally Var-ROR.
\end{proof}

For a formula in MU(1), there always exists a complete variable-disjunctive splitting. Furthermore, a complete variable-disjunctive splitting tree can be computed in polynomial time. Thus, every formula in MU(1) has a Var-ROR refutation that can be computed in polynomial time.\comment{ Since every minimal unsatisfiable Horn formula is in MU(1), every unsatisfiable Horn formula has a Var-ROR refutation, which can be computed in polynomial time. }

\begin{corollary}
Every formula in MU($1$) has a Var-ROR refutation that can be computed in polynomial time.
\end{corollary}

\comment{\begin{corollary}
Every unsatisfiable Horn formula has a Var-ROR refutation that can be constructed in polynomial time. 	
\end{corollary}}

Next, we will show that determining if a $2$CNF formula has a Var-ROR refutation is {\bf NP-complete}. 
It can easily be seen that this problem is in {\bf NP}. If the formula has $n$ variables, then at most $n$ resolution steps can be performed. 
The {\bf NP-hardness} will be shown by a reduction from the vertex-disjoint path problem for directed graphs.

\begin{definition} 
Given a directed graph $G$ and pairwise distinct vertexes $s_1, t_1, s_2, t_2$, the {\bf vertex-disjoint path problem} (2-DPP)
consists of finding a pair of vertex-disjoint paths in $G$, one from $s_1$ to $t_1$ and the other from $s_2$ to $t_2$.
\end{definition}

The problem is known to be {\bf NP-complete} \cite {FHW80}. Now we modify the problem as follows.

\begin{definition}
Given a directed graph $G$ and two distinct vertexes $s$ and $t$, the {\bf vertex-disjoint cycle problem} (C-DPP) 
consists of finding a pair of vertex-disjoint paths in $G$, one from $s$ to $t$ and the other from $t$ to $s$.
\end{definition}

Note that the paths are vertex-disjoint, if the inner vertexes of the path from $s$ to $t$ are disjoint from 
the inner vertexes of the path from $t$ to $s$.

\begin{lemma} 
	C-DPP is {\bf NP-complete}.
\end{lemma}

\begin{proof}
	Obviously, the problem is in {\bf NP}. We will show {\bf NP-hardness} by a reduction from 2-DPP.
	
	From $G=(V,E)$, $s_1$, $t_1$, $s_2$, and $t_2$ we construct the new graph
	\[G'=(V \cup\{s,t\}, E \cup\{(s,s_1), (t_2, s), (t_1, t), (t, s_2)\}).\]
	
	Assume that $G$ has two vertex-disjoint paths, $w_1$ from $s_1$ to $t_1$, and $w_2$ from $s_2$ to $t_2$.
	Thus, the paths $(s, s_1), w_1, (t_1,t)$ and $(t,s_2), w_2, (t_2, s)$ in $G'$ are vertex-disjoint. Note that $s_1, s_2, t_1, t_2$ are pairwise distinct. 
	Thus, $G'$ has the desired vertex-disjoint cycle.
	
	Now assume that $G'$ has two vertex-disjoint paths, $w_1$ from $s$ to $t$, and $w_2$ from $t$ to $s$.
	By construction, $w_1$ must contain a path from $s_1$ to $t_1$. Similarly, $w_2$ must contain a path from $s_2$ to $t_2$. Since $w_1$ and $w_2$ are vertex-disjoint these new paths must also be vertex-disjoint.
	Thus, $G$ has the desired vertex-disjoint paths.
	\end{proof}

\begin{theorem}
\label{varnp}
	Determining if a $2$CNF formula has a Var-ROR refutation is {\bf NP-complete}.
\end{theorem}

\begin{proof}
	As previously stated, we only need to show {\bf NP-hardness}. That will be done by a reduction from C-DPP.
	
	From $G=(V, E)$, $s$, and $t$ we construct a formula $\Phi$ in $2$CNF as follows:
	\begin{enumerate}
	\item For each vertex $v_i \in V - \{s,t\}$, create the variable $x_i$.
	\item Create the variable $x_0$.
	\item Let $v_i,v_j \in V - \{s,t\}$.
	\begin{enumerate}
	\item If $(s,v_i) \in E$ add the clause $(x_0 \rightarrow x_i)$ to $\Phi$.
	\item If $(t,v_i) \in E$ add the clause $(\neg x_0 \rightarrow x_i)$ to $\Phi$.
	\item If $(v_i,s) \in E$ add the clause $(x_i \rightarrow x_0)$ to $\Phi$.
	\item If $(v_i,t) \in E$ add the clause $(x_i \rightarrow \neg x_0)$ to $\Phi$.
	\item If $(v_i,v_j) \in E$ add the clause $(x_i \rightarrow x_j)$ to $\Phi$.
	\end{enumerate}
	\end{enumerate}
	
	 Assume that $G$ has two vertex-disjoint paths, 
	 \[w_1 = s, v_{i_1}, \ldots, v_{i_j}, t \text{ and } w_2 = t, v_{i_{j+1}}, \ldots, v_{i_k}, s.\]
	  Thus, there exist $2$CNF formulas $\Phi_1$ and $\Phi_2$ such that:
	 \begin{eqnarray*}
	 \Phi_1 & = & \{(x_0 \rightarrow x_{i_1}),(x_{i_1} \rightarrow x_{i_2}), \ldots, (x_{i_j} \rightarrow \neg x_0)\} \\ 
	 \Phi_2 & = & \{(\neg x_0 \rightarrow x_{i_{j+1}}),(x_{i_{j+1}} \rightarrow x_{i_{j+2}}), \ldots, (x_{i_k} \rightarrow x_0)\}.
	 \end{eqnarray*} 
	 
	 Clearly, $\Phi_1 \varres{} (\neg x_0)$ and $\Phi_2 \varres{} (x_0)$.
	Note that $x_0$ has not been used as a matching variable. Since $w_1$ and $w_2$ are vertex-disjoint, we have that 
	\[\{x_{i_1}, \ldots, x_{i_j}\} \cap \{x_{i_{j+1}}, \ldots, x_{i_k}\} = \emptyset.\]
	Thus, $\Phi_1 \cup \Phi_2 \varres{} \sqcup$. This means that $\Phi \supseteq \Phi_1 \cup \Phi_2$ is in Var-ROR.
	
	Now assume that $\Phi$ is in Var-ROR.
	Let $\Phi' \subseteq \Phi$ be minimally Var-ROR. We have that $\Phi'$ contains clauses with $x_0$ and $\neg x_0$. Otherwise,
	the formula would be satisfiable by setting each $x_i$ to {\bf true}.
	
	We proceed by an induction on the number of clauses in $\Phi'$.
	
	The shortest formula is $\Phi' = (x_0 \rightarrow \neg x_0) \wedge (\neg x_0 \rightarrow x_0)$. This $\Phi'$ is generated when $(s,t),(t,s) \in E$. These edges form the desired vertex-disjoint paths.
	
	Let $y$ be the variable for which $(y) \wedge (\neg y) \res{1} \sqcup$ is the last resolution step in $\Phi' \varres{} \sqcup$.
	Thus, $\Phi'$, can be divided into two variable-disjoint sets of clauses, $\Phi'_1$ and $\Phi'_2$, such that $\Phi'_1 \varres{} (y)$ and $\Phi'_2 \varres{} (\neg y)$. Otherwise, a variable would be used twice in $\Phi' \varres{} \sqcup$.
	
	Let $(L \vee x_i) \wedge (\neg x_i \vee K) \res{1} (L \vee K)$ a resolution step in
	$\Phi'_1 \varres{} (y)$ such that $(L \vee x_i) \in \Phi'_1$ and $(\neg x_i \vee K) \in \Phi'_1$. Thus, no clause with $x_i$ occurs in $\Phi'_2$ or $\Phi_1$ (except $(L \vee x_i)$ and $(\neg x_i \vee K)$).
	 Moreover, we see that the formula 
	 \[(\Phi' \setminus \{(L \vee a_i), (\neg a_i \vee K)\}) \cup \{(L \vee K)\}\]
	  is in Var-ROR.
	This formula represents the reduced graph where the edges $L \rightarrow a_i$ and $a_i \rightarrow K$ are replaced with the edge $L \rightarrow K$. By the induction hypothesis, there exists a vertex-disjoint cycle in this reduced graph. Thus, a vertex-disjoint cycle exists in $G$.
	\end{proof}

For arbitrary formulas in CNF, the problem of deciding whether a formula $\Phi$ has a sub-formula $\Phi'$ such that 
$\Phi' \in MU(1)$ is known to be  {\bf NP-complete}. But it was only known for arbitrary CNF.
Based on the Theorems above,  we obtain as a corollary that the MU(1) sub-formula problem is {\bf NP-complete} for $2$CNF, too.

\begin{corollary}
	The problem of deciding whether a formula in $2$CNF contains a minimal unsatisfiable formula with deficiency 1 is {\bf NP-complete}.
\end{corollary}

\section{The Complexity of ROR for 2CNF}
\label{for}
In this section, we show that determining if a formula in $2$CNF has a read-once refutation is {\bf NP-complete}. It was established in \cite{Sze01} that the Var-ROR problem for $2$CNF can be reduced to the ROR problem for $2$CNF. Together with Theorem \ref{varnp}, this establishes that the ROR problem for $2$CNF is {\bf NP-complete}. We now present an alternate way of obtaining this result.

Unlike minimally Var-ROR formulas, minimally ROR formulas are not necessarily minimal unsatisfiable. They also can have deficiencies other than 1. An example of such a formula can be seen in Appendix \ref{noMU}.

We now prove some properties of minimal unsatisfiable formulas in 2-CNF with one or two unit clauses.
It can easily be seen that such formulas contain at most two unit clauses.
\begin{lemma} \label{1mu}
	Let $\Phi$ be a minimal unsatisfiable $2$CNF formula.
	\begin{enumerate}
		\item If $\Phi$ contains two unit clauses, then $\Phi$ has the form
		\[(L), (\neg L \vee L_1), \ldots,  (\neg L_{t-1} \vee L_t), (\neg L_t \vee \neg K), (K)\] where $L, L_1, \ldots, L_t, K$ are pairwise distinct.
		\item If $\Phi$ contains exactly one unit clause, then $\Phi$ has the form
		\begin{eqnarray*}
		(L), (\neg L \vee L_1), (\neg L_1 \vee L_2), \ldots, (\neg L_t \vee K), \\
		(\neg K \vee S_1), (\neg S_1 \vee S_2), \ldots, (\neg S_q \vee R), \\
		(\neg K \vee P_1), (\neg P_1 \vee P_2), \ldots, (\neg P_m \vee \neg R)
		\end{eqnarray*}
		where the literals are all pairwise distinct.
		\item If $\Phi$ contains at least one unit clause, then $\Phi$ has a read-once resolution refutation.
		\item If $\Phi$ is in MU($1$), then $\Phi$ has a ROR refutation.
	\end{enumerate}	
	
\end{lemma} 	
\begin{proof}
We prove each part of the lemma separately.

The proofs of part 1 and 2 are straightforward because no minimal unsatisfiable $2$CNF formula contains more than two unit clauses.
\begin{enumerate}
  \setcounter{enumi}{2}
\item If the formula has two unit clauses, then structure of the formula leads immediately to a ROR refutation.

 If the formula has one unit clause, then we can perform the desired resolution refutation as follows:
 \begin{enumerate}
	\item First, we resolve  $ (\neg K \vee S_1), (\neg S_1 \vee S_2), \ldots, (\neg S_q \vee R)$ to obtain $(\neg K \vee R)$.
	\item Then, we resolve
	$(\neg K \vee P_1), (\neg P_1 \vee P_2), \ldots, (\neg P_m \vee \neg R)$
	to obtain $(\neg K \vee \neg R)$.
	\item Next, we perform the resolution step
	\[(\neg K \vee R) \wedge (\neg K \vee \neg R) \res{1} (\neg K).\]
	 \item Finally, the unit clause $(\neg K)$ together with the chain
	\[(L), (\neg L \vee L_1), (\neg L_2 \vee L_3), \ldots, (\neg L_t \vee K)\]
	 resolve to finish the ROR refutation.
	\end{enumerate}
	\item Every $2$CNF formula in MU($1$) has a complete
	disjunctive splitting tree. This guarantees the existence of a ROR refutation \cite{KBZ02}. 
\end{enumerate}
\end{proof}

\begin{theorem}
\label{rorsplit}
	Let $\Phi$ be in $2$CNF.
	$\Phi$ is in ROR, if and only if there exists a sub-formula $\Phi' \subseteq \Phi$ for which there exists a variable $x$ and a 
	disjunctive splitting $(F_x, F_{\neg x})$ over $x$, such that $F_x, F_{\neg x}$ are in MU(1).
\end{theorem}

\begin{proof}
	Suppose, there exists a sub-formula $\Phi' \subseteq \Phi$ with disjunctive splitting $(F_x, F_{\neg x})$,
	where $F_x, F_{\neg x} \in $ MU(1). We have that $F_x$ and $F_{\neg x}$ each contain at least one unit clause. Now we reconstruct the clauses of 
	$F_x$ and $F_{\neg x}$ by adding the removed literal $x$ (resp. $\neg x$) to the clauses in $F_x$ (resp. $F_{\neg x}$). These new formulas are denoted as $F_x^x$ and $F_{\neg x}^{\neg x}$.

	From lemma \ref{1mu}, every formula in MU(1) with a unit-clause has a read-once resolution refutation. We also have that $(x)$ and $(\neg x)$ do not occur in the splitting formulas. Thus, we get $F_x \rores{} \sqcup$, $F_{\neg x} \rores{} \sqcup$, $F_x^x \rores{} (x)$, and $F_{\neg x}^{\neg x} \rores{} (\neg x)$. Now we have to guarantee that there is a read-once resolution for $\Phi$.
	$(F_x, F_{\neg x})$ is disjunctive splitting. Thus, no clause of $\Phi$ occurs in both $F_x^x$ and in $F_{\neg x}^{\neg x}$. We can combine the resolutions $F_x^x \rores{} (x)$ and $F_{\neg x}^{\neg x} \rores{} (\neg x)$, with the resolution step $(x) \wedge (\neg x) \res{1} \sqcup$ to yield $\Phi \rores{} \sqcup$, since $F_x^x,F_{\neg x}^{\neg x} \subseteq \Phi$.

	Now suppose, $\Phi \in $ ROR and without loss of generality is minimally ROR. We will show that $\Phi$ contains the desired splitting. Let 
	$x \wedge \neg x \res{} \sqcup$ the last resolution step in the read-once resolution refutation. Furthermore,
	let $F_x^x$ ($F_{\neg x}^{\neg x}$ respectively) be the set of original clauses from $\Phi$ used in the derivation of $x$ ($\neg x$ respectively). These sets have no clause in common because together they form a read-once resolution refutation for $\Phi$.
	
	 The formulas have the form
	\[F_x^x = (x \vee L_1) \wedge \ldots \wedge  (x \vee L_t) \wedge \sigma_1 \text{ and } F_{\neg x}^{\neg x} = (\neg x \vee K_1) \wedge \ldots \wedge  (\neg x \vee K_r) \wedge \sigma_2,\]
	 where $\sigma_1 \cap \sigma_2 = \emptyset$.
	
	Thus, we can construct the formulas 
	\[F_x = (L_1) \wedge \ldots \wedge  (L_t) \wedge \sigma_1 \text{ and } F_{\neg x} = (K_1) \wedge \ldots \wedge  (K_r) \wedge \sigma_2,\]
	 where $\sigma_1 \cap \sigma_2 = \emptyset$.
	 
	By construction, $F_x \rores{} \sqcup$ and $F_{\neg x} \rores{} \sqcup$. Both $F_x$ and $F_{\neg x}$ are minimal unsatisfiable. Otherwise, $\Phi$ would not be minimally ROR. Thus, $(F_x,F_{\neg x})$ is a disjunctive splitting. Both $F_x$ and $F_{\neg x}$ contain unit clauses. Thus, by lemma \ref{lem01}, $F_x$ and $F_{\neg x}$ are MU($1$).
\comment{	 If the formula $F_x$ contains a minimal unsatisfiable sub-formula with some of the unit clauses $L_i$, then the sub-formula is in MU(1), see lemma \ref{1mu}. If additionally the formula $F_{\neg x}$ contains a minimal sub-formulas with a unit clause $K_j$,
	then the sub-formula is in MU(1), too.
	Now we add to the sub-formulas the omitted literals $x$ and $\neg x$. The resulting formulas describe the desired splitting 
	into two disjunctive MU(1) formulas.\\
	Suppose, one of the formulas $F_x$ or $F_{\neg x}$ contains no minimal unsatisfiable formula which includes a unit clause
	generated by omitting $x$ resp. $\neg x$, say that is the case for $F_x$. Thus, every minimal unsatisfiable sub-formula 
	$H \subseteq F_x$ has no unit clause $L_1, \ldots, L_t$. If $H$ contains other unit clauses, then $H$ is in MU(1) and we obtain our
	desired splitting providing that is the case for $F_{\neg x}$, too.\\
	Thus, it remains to discuss the case, that every minimal unsatisfiable sub-formula $H \subseteq F_x$ and every minimal unsatisfiable 
	sub-formula $T \subseteq F_{\neg x}$ contain no unit-clause.\\
	But we there is read once resolution refutation $F_x^x \rores{} x$. Therefore, for some of the literals $L_{i_t}$, there is
	a resolution path from $(x \vee L_{i_t})$ to some of the clauses in $H$. Without the literal $x$ that means there is a resolution path from $L_{i_j}$ to some clause in $H$ using only clauses in $F_x$. In that case, $F_x$ contains a minimal unsatisfiable formula with 
	unit clause $L$ in contradiction to our assumption.	}
\end{proof}

We will now prove the {\bf NP-completeness} or the ROR problem for $2$CNF formulas.
Instead of using the vertex-disjoint cycle problem, we will be reducing from the edge-disjoint cycle problem for directed graphs.

\begin{definition}
Given a directed graph $G$ and two distinct vertexes $s$ and $t$, the {\bf edge-disjoint cycle problem} (C-DEP) 
consists of finding a pair of edge-disjoint paths in $G$, one from $s$ to $t$ and the other from $t$ to $s$.
\end{definition}

The problem is {\bf NP-complete}. For two pairs of vertexes, the edge-disjoint path problem
is {\bf NP-complete} \cite{EIS76}. We can reduce the edge-disjoint path problem to C-DEP the same way we reduced 2-DPP to C-DPP.

\begin{theorem}
	The ROR problem for $2$CNF formulas is {\bf NP-complete}.
\end{theorem}

\begin{proof}
	ROR is in {\bf NP} for arbitrary formulas in CNF \cite{IM95}. Thus, we only need to show {\bf NP-hardness}. That will be done by a reduction from C-DEP.
	
	From $G=(V, E)$, $s$, and $t$ we construct a formula $\Phi$ in $2$CNF as follows:
	\begin{enumerate}
	\item For each vertex $v_i \in V - \{s,t\}$, create the variable $x_i$.
	\item Create the variable $x_0$.
	\item Let $v_i,v_j \in V - \{s,t\}$.
	\begin{enumerate}
	\item If $(s,v_i) \in E$ add the clause $(x_0 \rightarrow x_i)$ to $\Phi$.
	\item If $(t,v_i) \in E$ add the clause $(\neg x_0 \rightarrow x_i)$ to $\Phi$.
	\item If $(v_i,s) \in E$ add the clause $(x_i \rightarrow x_0)$ to $\Phi$.
	\item If $(v_i,t) \in E$ add the clause $(x_i \rightarrow \neg x_0)$ to $\Phi$.
	\item If $(v_i,v_j) \in E$ add the clause $(x_i \rightarrow x_j)$ to $\Phi$.
	\end{enumerate}
	\end{enumerate}
	
	 Assume that $G$ has two edge-disjoint paths,
	 \[w_1 = s, v_{i_1}, \ldots, v_{i_j}, t \text{ and } w_2 = t, v_{i_{j+1}}, \ldots, v_{i_k}, s.\]
	  Thus, there exist $2$CNF formulas $\Phi_1$ and $\Phi_2$ such that:
	 \begin{eqnarray*}
	 \Phi_1 & = & \{(x_0 \rightarrow x_{i_1}),(x_{i_1} \rightarrow x_{i_2}), \ldots, (x_{i_j} \rightarrow \neg x_0)\} \\ 
	 \Phi_2 & = & \{(\neg x_0 \rightarrow x_{i_{j+1}}),(x_{i_{j+1}} \rightarrow x_{i_{j+2}}), \ldots, (x_{i_k} \rightarrow x_0)\}.
	 \end{eqnarray*} 
	 
	 Obviously, $\Phi_1 \rores{} (\neg x_0)$ and $\Phi_2 \rores{} (x_0)$.
	Note that $x_0$ has not been used as a matching variable. Since $w_1$ and $w_2$ are edge-disjoint, we have that $\Phi_1 \cap \Phi_2 = \emptyset$.
	Thus, $\Phi_1 \cup \Phi_2 \rores{} \sqcup$. This means that $\Phi \supseteq \Phi_1 \cup \Phi_2$ is in ROR.
	
	Now assume that $\Phi$ is in ROR. \\
	Let $\Phi' \subseteq \Phi$ be minimally ROR. We have that $\Phi'$ contains clauses with $x_0$ and $\neg x_0$. Otherwise,
	the formula would be satisfiable by setting each $x_i$ to {\bf true}.
	
	We proceed by an induction on the number of clauses in $\Phi'$.
	
	The shortest formula is $\Phi' = (x_0 \rightarrow \neg x_0) \wedge (\neg x_0 \rightarrow x_0)$. This $\Phi'$ is generated when $(s,t),(t,s) \in E$. These edges form the desired edge-disjoint paths.
	
	Let $(L \rightarrow K) \wedge (K \rightarrow R) \res{1} (L \vee R)$ be a resolution step where $(L \rightarrow K) \in \Phi'$ and $(K \rightarrow R) \in \Phi'$. Note that $( L \rightarrow R) \not \in \Phi'$. Otherwise, $\Phi'$ would not be minimally ROR. 
	
	In a read-once refutation, we remove the parent clauses from $\Phi$ and add the resolvent $(L \rightarrow R)$.
	This new formula has a read-once resolution refutation and can be considered as obtained by a reduced graph without the edges $ L \rightarrow K, K \rightarrow R$ but with the edge $L \rightarrow R$. By the induction hypothesis, this new graph contains the desired edge-disjoint cycle. If we replace the edge $L \rightarrow R$ in this cycle with $L \rightarrow K$ and $K \rightarrow R$, then we construct the desired edge-disjoint cycle in $G$.
\end{proof}

By lemma \ref{lem04}, the problem of determining if an MU-formula in $2$CNF has a disjunctive splitting whose splitting formulas are in MU(1) can be decided in
polynomial time. Hence, for MU-formulas in $2$CNF, the ROR problem is solvable in polynomial time.

\begin{corollary}
	The ROR problem for minimal unsatisfiable $2$CNF formulas is in {\bf P}.
\end{corollary}

\comment{
\begin{lemma}
	Every unsatisfiable Horn formula has a ROR refutation, which can be computed in polynomial time.
\end{lemma}

\begin{proof}
Let $\Phi$ be an unsatisfiable Horn formula. Since the satisfiability problem for Horn formulas can be solved in polynomial time, we can compute a minimal unsatisfiable sub-formula, say $\Phi'$, in polynomial time. The formula is in MU(1) and has therefore a ROR refutation. Every minimal unsatisfiable Horn formula contains as least one positive unit clause. Let 
$\Phi' = (x) \wedge (\neg x \vee \pi_1) \wedge \ldots \wedge (\neg x \vee \pi_m) \wedge \sigma$, where in $\sigma$ the variable $x$ do not occur. The formula $\Phi'_x = \pi_1 \wedge \ldots \pi_m \wedge \sigma$ is minimal unsatisfiable
and in MU(1). Suppose, we have a read-once resolution refutation 
$\Phi'_x \rores{} \sqcup$. Then, we get a read-once resolution derivation  $(\neg x \vee \pi_1) \wedge \ldots \wedge(\neg x \vee \pi_m) \wedge \sigma \rores{} \neg x$ and finally with the unit clause $(x)$ a ROR refutation $\Phi \rores{} \sqcup$.\\
That leads to the desired procedure: Forward we apply unit propagation with
$\Phi$ and backwards we construct the ROR refutation. In a certain sense,
that looks like generating an input resolution refutation given a unit resolution refutation.
\end{proof}	
}

\section{Copy Complexity of 2CNF formulas}
\label{copy}

Let $\Phi$ be an unsatisfiable formula in CNF and let $\Lambda$ be a tree-like resolution refutation of $\Phi$. For each clause $\pi_i \in \Phi$, let $\Lambda_i$ be the number of times $\pi_i$ is used in a resolution step of $\Lambda$.

\begin{definition}
	A CNF formula $\Phi$ has a {\bf Copy Complexity} of $k$, if there exists $\Lambda$ such that $\Lambda_i \le k$ for $i=1, \ldots, m$.
\end{definition}

Note that in a tree-like refutation reusing a resultant clause requires the reuse of the clauses originally in $\Phi$. Thus, by limiting the number of times $\Lambda$ can use each clause in $\Phi$ we also limit the number of times each resultant clause can be used.

We can equivalently define copy complexity as follows.

\begin{definition}
	A CNF formula $\Phi$ has a {\bf Copy Complexity} of $k$, if there exists a multi-set of CNF clauses, $\Phi'$ such that:
	\begin{enumerate}
		\item Every clause in $\Phi$ appears at most $k$ times in $\Phi'$.
		\item Every clause in $\Phi'$ appears in $\Phi$.
		\item $\Phi'$ has a read-once refutation.
	\end{enumerate}
\end{definition}

Thus, if a formula, $\Phi$, has a copy complexity of $1$ then the formula has a read-once resolution refutation.

We can extend the concept of copy complexity to classes of CNF formulas.

\begin{definition}
	A class of CNF formulas has a {\bf Copy Complexity} of $k$, if every formula in that class has a copy complexity of $k$.
\end{definition}

We will now show that $2$CNF has a copy complexity of $2$. This means that we can always prove unsatisfiablity using each clause at most twice.

\begin{lemma}
	\label{rolem}
	Let $\Phi$ be a formula in $2$CNF. If we can prove $(x_i)$, then we can prove $(x_i)$ using each clause no more than once.
\end{lemma}

\begin{proof}
	For a formula $\Phi$, we construct the corresponding implication graph $G$ as follows:
	\begin{enumerate}
	\item For each variable $x_i$, we create the verticies $x_i$ and $\bar{x}_i$. These correspond to the literals $x_i$ and $\neg x_i$.
	\item For each clause $(L \vee K)$, we create the edges edges
	$\bar{L} \rightarrow K$ and $\bar{K} \rightarrow L$.
	\item For each unit clause $(L)$, we create the edge $\bar{L} \rightarrow L$.
	\end{enumerate}
	
	 We know that we can prove $(x_i)$ in $\Phi$, if and only if there exists a path from $\bar{x}_i$ to $x_i$ in $G$. Let $p$ denote this path. If no two edges in $p$ correspond to the same clause, then this path already corresponds to a read-once proof of $(x_i)$.
	
	Let $e$ be the first edge in $p$ such that the other edge corresponding to the same clause as $e$ has already been used in $p$. We can assume without loss of generality that $e$ is $x_j \rightarrow x_k$. Thus, we can break $p$ up as follows:
	\begin{enumerate}
		\item a path, $p_1$, from $\bar{x}_i$ to $\bar{x}_k$,
		\item the edge $\bar{x}_k \rightarrow \bar{x}_j$ (the edge corresponding to the same clause as $e$),
		\item a path, $p_2$, from $\bar{x}_j$ to $x_j$,
		\item the edge $x_j \rightarrow x_k$ (the edge $e$),
		\item and a path, $p_3$, from $x_k$ to $x_i$.
	\end{enumerate}

	This can be seen in Figure \ref{imppath}.
	
	\begin{center}
		\begin{figure}[htb]
			\centering
			\begin{tikzpicture}[scale=1]
			\path (0,0) node[draw,shape=circle] (xi) {$x_i$};
			\path (3,-3) node[draw,shape=circle] (xk') {$\bar{x}_k$};
			\path (6,-3) node[draw,shape=circle] (xj') {$\bar{x}_j$};
			\path (6,0) node[draw,shape=circle] (xj) {$x_j$};
			\path (3,0) node[draw,shape=circle] (xk) {$x_k$};
			\path (0,-3) node[draw,shape=circle] (xi') {$\bar{x}_i$};
			
			\draw[-latex] (xi') -- (xk');
			\draw[-latex] (xk') -- (xj');
			\draw[-latex] (xj') -- (xj);
			\draw[-latex] (xj) -- (xk);
			\draw[-latex] (xk) -- (xi);
			
			\path (1.5,-3.25) node {$p_1$};
			\path (6.25,-1.5) node {$p_2$};
			\path (1.5,.25) node {$p_3$};
			\path (4.5,.25) node {$e$};
			\end{tikzpicture}
			\caption{Example of Path $p$}
			\label{imppath}
		\end{figure}
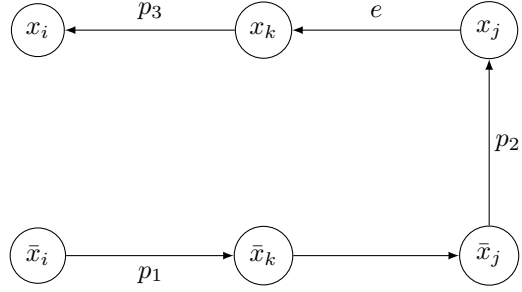
	\end{center}
	
	By our choice of $e$, we know no two edges in $p_1$ or $p_2$ correspond to the same clause.
	Thus, $p_2$ corresponds to a a read-once proof of $(x_j)$. We also have that $p_1$ combined with the edge $\bar{x}_k \rightarrow \bar{x}_j$ is a read once proof of $(\neg x_i \rightarrow \neg x_j)$. Combining these two yields a read-once proof of $(\neg x_i \rightarrow x_i)$ and therefore for $(x_i)$.
\end{proof}

\begin{theorem}
	If a formula in $2$CNF is unsatisfiable, then there exists a resolution refutation in which no clause is used more than twice.
\end{theorem}

\begin{proof}
	Suppose, the formula $\Phi$ in $2$CNF is unsatisfiable. Then for some $x_i$, we have that:
	\begin{enumerate}
		\item $\Phi \res{} (x_i)$
		\item and $\Phi \res{} (\neg x_i)$.
	\end{enumerate}
	From Lemma \ref{rolem}, we know that we can derive $(x_i)$ using each clause at most once. Similarly we can derive $(\neg x_i)$ using each clause at most once.
	
	Thus, for $\Phi$, there is a resolution refutation using each clause at most twice. That is once in the refutation derivation of $(x_i)$ and once in the resolution derivation of $(\neg x_i)$.
\end{proof}
\section{Conclusion}
\label{conc}

	 In this paper, we discuss the computational complexity of ROR decidability in a well-known classe of
	 CNF formulas, viz., $2$CNF. Prior research had established that the ROR decidability problem is {\bf NP-complete}
	 for $3$CNF formulas. Likewise, it is known that resolution refutations for $3$CNF formulas, even when they exist need not
	 be polynomial sized. For $2$CNF formulas, it is well-known that polynomial sized refutations exist in the
	 general case. We showed that ROR decidability in
	 $2$CNF formulas is {\bf NP-complete}. Additionally, we showed that the  copy complexity of $2$CNF formulas is $2$.
	 This means that every unsatisfiable $2$CNF formula has  a refutation in which each input clause is used at most twice.
	 Furthermore, the optimal length resolution when clause copy is permitted can be  determined in polynomial time.

\newpage
\appendix
\section{ROR Refutation Example}
\label{refexamp}
We now apply read-once resolution refutation to generate a refutation of the $2$SAT instance specified by Formula (\ref{exsys}).
\begin{eqnarray}
\label{exsys}
(x_1 ,x_2) & (\neg x_1, x_3) & (\neg x_1, x_4) \nonumber\\
(\neg x_2, x_3) & (\neg x_2, x_4) & (\neg x_3, x_5) \\
(\neg x_3 , x_6) & (\neg x_4, \neg x_5) & (\neg x_4, \neg x_6) \nonumber
\end{eqnarray}

The application of this read-once resolution refutation to Formula (\ref{exsys}) can be seen in Figure \ref{exfig1}.

\begin{figure}[H]
\footnotesize
\center
\begin{tikzpicture}[scale=1]

\draw (-6.4,1) -- (-8,1) -- (-8,1.5) -- (-6.4,1.5) -- (-6.4,1);
\path (-7.2,1.25) node {$(\neg x_3, x_5)$};

\draw (-6.4,0) -- (-8,0) -- (-8,.5) -- (-6.4,.5) -- (-6.4,0);
\path (-7.2,.25) node {$(\neg x_4, \neg x_5)$};

\draw (-6.4,-1) -- (-8,-1) -- (-8,-.5) -- (-6.4,-.5) -- (-6.4,-1);
\path (-7.2,-.75) node {$(\neg x_1, x_3)$};

\draw (-6.4,-2) -- (-8,-2) -- (-8,-1.5) -- (-6.4,-1.5) -- (-6.4,-2);
\path (-7.2,-1.75) node {$(\neg x_1, x_4)$};

\draw (-6.4,-3) -- (-8,-3) -- (-8,-2.5) -- (-6.4,-2.5) -- (-6.4,-3);
\path (-7.2,-2.75) node {$(\neg x_3, x_6)$};

\draw (-6.4,-4) -- (-8,-4) -- (-8,-3.5) -- (-6.4,-3.5) -- (-6.4,-4);
\path (-7.2,-3.75) node {$(\neg x_4, \neg x_6)$};

\draw (-6.4,-5) -- (-8,-5) -- (-8,-4.5) -- (-6.4,-4.5) -- (-6.4,-5);
\path (-7.2,-4.75) node {$(\neg x_2, x_3)$};

\draw (-6.4,-6) -- (-8,-6) -- (-8,-5.5) -- (-6.4,-5.5) -- (-6.4,-6);
\path (-7.2,-5.75) node {$(\neg x_2, x_4)$};

\draw (-6.4,-7) -- (-8,-7) -- (-8,-6.5) -- (-6.4,-6.5) -- (-6.4,-7);
\path (-7.2,-6.75) node {$(x_1, x_2)$};

\draw (-4.6,0.5) -- (-6.2,0.5) -- (-6.2,1) -- (-4.6,1) -- (-4.6,0.5);
\path (-5.4,0.75) node {$(\neg x_3, \neg x_4)$};
\draw[-latex] (-6.4,1.25) -- (-5.4, 1);
\draw[-latex] (-6.4,0.25) -- (-5.4, 0.5);

\draw (-4.6,-3.5) -- (-6.2,-3.5) -- (-6.2,-3) -- (-4.6,-3) -- (-4.6,-3.5);
\path (-5.4,-3.25) node {$(\neg x_3, \neg x_4)$};
\draw[-latex] (-6.4,-2.75) -- (-5.4, -3);
\draw[-latex] (-6.4,-3.75) -- (-5.4, -3.5);

\draw (-2.8,0) -- (-4.4,0) -- (-4.4,0.5) -- (-2.8,0.5) -- (-2.8,0);
\path (-3.6,0.25) node {$(\neg x_1, \neg x_4)$};
\draw[-latex] (-4.6,0.75) -- (-3.6, 0.5);
\draw[-latex] (-6.4,-0.75) -- (-3.6, 0);

\draw (-2.8,-4) -- (-4.4,-4) -- (-4.4,-3.5) -- (-2.8,-3.5) -- (-2.8,-4);
\path (-3.6,-3.75) node {$(\neg x_2, \neg x_4)$};
\draw[-latex] (-4.6,-3.25) -- (-3.6, -3.5);
\draw[-latex] (-6.4,-4.75) -- (-3.6, -4);

\draw (-1,-0.5) -- (-2.6,-0.5) -- (-2.6,0) -- (-1,0) -- (-1,-0.5);
\path (-1.8,-.25) node {$(\neg x_1)$};
\draw[-latex] (-2.8,.25) -- (-1.8, 0);
\draw[-latex] (-6.4,-1.75) -- (-1.8, -.5);

\draw (-1,-4.5) -- (-2.6,-4.5) -- (-2.6,-4) -- (-1,-4) -- (-1,-4.5);
\path (-1.8,-4.25) node {$(\neg x_2)$};
\draw[-latex] (-2.8,-3.75) -- (-1.8, -4);
\draw[-latex] (-6.4,-5.75) -- (-1.8, -4.5);

\draw (.8,-5) -- (-.8,-5) -- (-.8,-4.5) -- (.8,-4.5) -- (.8,-5);
\path (0,-4.75) node {$(x_1)$};
\draw[-latex] (-1,-4.25) -- (0, -4.5);
\draw[-latex] (-6.4,-6.75) -- (0, -5);

\draw (1,-3) -- (2.6,-3) -- (2.6,-2.5) -- (1,-2.5) -- (1,-3);
\path (1.8,-2.75) node {$\emptyset$};
\draw[-latex] (-1,-.25) -- (1.8, -2.5);
\draw[-latex] (.8,-4.75) -- (1.8, -3);
\end{tikzpicture}
\caption{Read-Once Refutation}
\label{exfig1}
\end{figure}
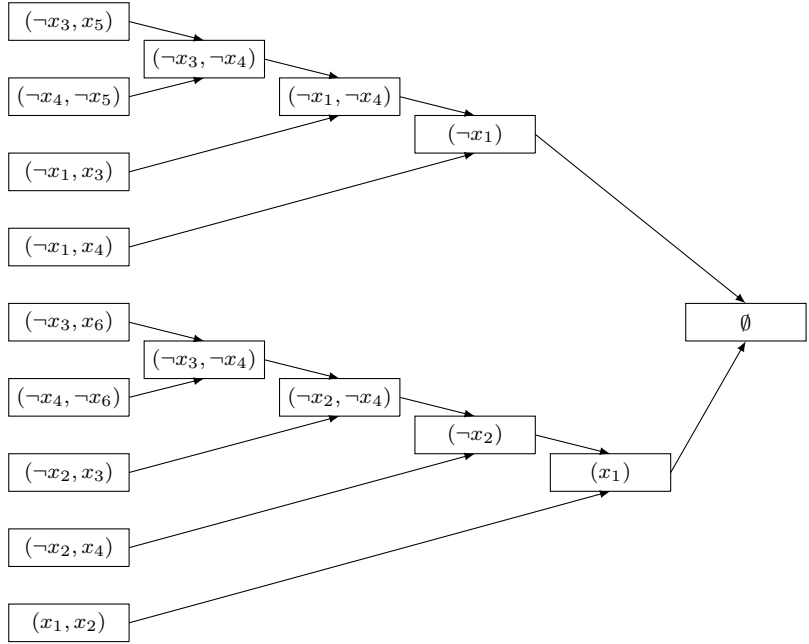

Note that the clause $(\neg x_3, \neg x_4)$ is used twice. However, this is still a read-once refutation since each time the clause $(\neg x_3, \neg x_4)$ is derived different clauses from the original formula are used.
\section{Non MU(1) ROR refutation}
\label{noMU}
\begin{example} 
Let 
\begin{eqnarray*}
\Phi & = & \{(\neg x \rightarrow a), (a \rightarrow b), (a \rightarrow c), (c \rightarrow b),\\
& & (b \rightarrow x), (b \rightarrow \neg x), (x \rightarrow a)\}.
\end{eqnarray*}
$\Phi$ is not minimal unsatisfiable and has deficiency 2.
Minimal unsatisfiable sub-formulas of $\Phi$ are:
\begin{eqnarray*}
\Phi_1 & = & \{(\neg x \rightarrow a), (a \rightarrow b), (b \rightarrow x), (b \rightarrow \neg x), (x \rightarrow a) \}\\
\Phi_2 & = & \{(\neg x \rightarrow a), (a \rightarrow c), (c \rightarrow b),\\
& & (b \rightarrow x), (b \rightarrow \neg x), (x \rightarrow a)\}.
\end{eqnarray*}
Note that $\Phi_1$ and $\Phi_2$ have deficiency 1. There is no read-once resolution refutation for $\Phi_1$ or $ \Phi_2$.
However, $\Phi$ has the following read-once resolution refutation:
\begin{enumerate}
\item $(\neg x \rightarrow a) \wedge (a \rightarrow c) \res{1} (\neg x \rightarrow c)$.

\item $(\neg x \rightarrow c) \wedge (c \rightarrow b) \res{1} (\neg x \rightarrow b)$.

\item $ (\neg x \rightarrow b) \wedge (b \rightarrow x) \res{1} (x)$.

\item $(x \rightarrow a) \wedge (a \rightarrow b) \res{1} (x \rightarrow b$).

\item$ (x \rightarrow b) \wedge (b \rightarrow \neg x) \res{1} (\neg x)$.

\item $(\neg x) \wedge (x) \res{1} \sqcup$.
\end{enumerate}
From $\Phi$ we can construct the graph $G$ as follows:
\begin{enumerate}
\item For each literal in $\Phi$, create a vertex in $G$.
\item For each implication $(x \rightarrow a)$ in $\Phi$, create the edge $(x \rightarrow a)$ in $G$.
\end{enumerate}
The result of this can be seen in Figure \ref{cycle}.

\begin{center}
	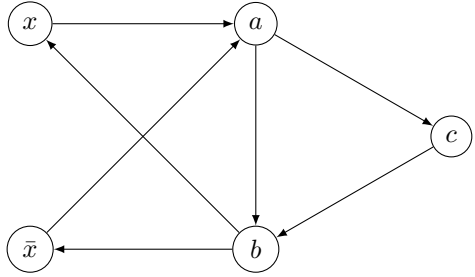
\begin{figure}[htb]
		\centering
		\begin{tikzpicture}[scale=1]
		\path (0,0) node[draw,shape=circle] (x) {$x$};
		\path (0,-3) node[draw,shape=circle] (x') {$\bar{x}$};
		\path (3,0) node[draw,shape=circle] (a) {$a$};
		\path (3,-3) node[draw,shape=circle] (b) {$b$};
		\path (5.6,-1.5) node[draw,shape=circle] (c) {$c$};
		
		\draw[-latex] (x) -- (a);
		\draw[-latex] (x') -- (a);
		\draw[-latex] (a) -- (b);
		\draw[-latex] (a) -- (c);
		\draw[-latex] (c) -- (b);
		\draw[-latex] (b) -- (x);
		\draw[-latex] (b) -- (x');
		\end{tikzpicture}
		\caption{Edge-Disjoint Cycle: $ x$ to $\bar{x}$ and $\bar{x}$ to $ x$}
		\label{cycle}
	\end{figure}
\end{center}

The paths $\neg x, a, b, x$ and $x, a, c, b, \neg x$ form an edge-disjoint cycle from $\neg x$ to $x$ and from $x$ to $\neg x$.

These paths correspond to the follow sub-formulas of $\Phi$:
\begin{eqnarray*}
F_x^x &=& (\neg x \rightarrow a) \wedge (b \rightarrow x) \wedge (a \rightarrow b) \\
F_{\neg x}^{\neg x} &=& (a \rightarrow c) \wedge (c \rightarrow b) \wedge (x \rightarrow a) \wedge ( b \rightarrow \neg x).
\end{eqnarray*}
\end{example}

The following summarizes our observations:
\begin{enumerate}
	\item The formula $\Phi$ is in ROR, but not minimal unsatisfiable. Moreover, any minimal unsatisfiable sub-formula of $\Phi$ has no read-once resolution refutation.
	\item The formula $\Phi$ is minimally ROR and $\Phi$ has deficiency greater than 1.
	\item The formula $\Phi$ has two sub-formulas formulas in MU(1). 
	\item The formula $\Phi$, when considered as directed graph, has an edge-disjoint cycle from $\neg x$ to $x$ and from $x$ to $\neg x$. Each path in  this cycle corresponds to a sub-formula of $\Phi$.
\end{enumerate}
\comment{
The last observation leads to a proof of the {\bf NP-hardness} in which we establish a reduction from the edge-disjoint path problem.
This is similar to the reduction from the vertex-disjoint path problem used to prove the {\bf NP-hardness} of Var-ROR.}

    }

\end{document}